\documentclass[11pt,leqno]{amsart}
\usepackage[english]{babel}
\usepackage{amsfonts,amssymb,graphicx,amscd,amsmath,amsthm,hyperref,tikz,booktabs,tabularx,multirow,multicol,stmaryrd}
\usepackage{enumitem}
\usepackage[flushleft]{threeparttable}
\usepackage[fixamsmath]{mathtools}
\allowdisplaybreaks

\newtheorem{theorem}{Theorem}[section]

\theoremstyle{definition}

\newtheorem{definition}[theorem]{Definition}

\newtheorem{problem}{Problem}
\newtheorem{remark}[theorem]{Remark}

\newcommand{\Z}{\mathbb Z}

\newcommand{\C}{\mathbb C}

\newcommand{\GL}{\operatorname{GL}}
\newcommand{\Sp}{\operatorname{Sp}}

\newcommand{\z}{\mathbf{z}}

\newcommand{\gammaice}[7]{
\begin{tikzpicture}[scale=0.80, every node/.style={scale=0.80}]
\coordinate (a) at (-.75, 0);
\coordinate (b) at (0, .75);
\coordinate (c) at (.75, 0);
\coordinate (d) at (0, -.75);
\coordinate (aa) at (-.75,.5);
\coordinate (cc) at (.75,.5);
\coordinate (ee) at (0,0);
\draw[semithick] (a)--(c);
\draw[semithick] (b)--(d);
\draw[fill=white, very thin] (a) circle (.25);
\draw[fill=white, very thin] (b) circle (.25);
\draw[fill=white, very thin] (c) circle (.25);
\draw[fill=white, very thin] (d) circle (.25);
\path[fill=white] (ee) circle (.30);
\node at (0,1) { };
\node at (a) {$#1$};
\node at (b) {$#2$};
\node at (c) {$#3$};
\node at (d) {$#4$};
\node at (aa) {$#5$};
\node at (cc) {$#6$};
\node at (ee) {$#7$};
\end{tikzpicture}}

\newcommand{\lhs}[9]{\raisebox{-49pt}{\resizebox{5.4cm}{4.5cm}{\begin{tikzpicture}
\path[use as bounding box](-1,-2) rectangle(5,3);
\coordinate (a) at (0,-1);     % 1
\coordinate (b) at (0,-1.5);   % R
\coordinate (c) at (0,1);      % 2
\coordinate (d) at (0,1.5);    % 
\coordinate (e) at (3,2);      % 3
\coordinate (f) at (4,1);      % 4
\coordinate (g) at (4,1.5);    % 
\coordinate (h) at (4,-1);     % 5
\coordinate (i) at (4,-1.5);   % 
\coordinate (j) at (3,-2);     % 6
\coordinate (k) at (2,1);      % 7
\coordinate (l) at (2,1.5);    % 
\coordinate (m) at (2,-1);     % 8
\coordinate (n) at (2,-1.5);   % 
\coordinate (o) at (3,0);      % 9
\coordinate (p) at (1,0);      % T
\coordinate (q) at (3,1);      % 
\coordinate (r) at (3,-1);     % 
\draw[semithick] (a) to [out=0, in=180] (k);
\draw[semithick] (k)--(f);
\draw[semithick] (c) to [out=0, in=180] (m);
\draw[semithick] (m)--(h);
\draw[semithick] (e)--(j);
\draw[fill=white, very thin] (a) circle (.25);
\draw[fill=white, very thin] (c) circle (.25);
\draw[fill=white, very thin] (e) circle (.25);
\draw[fill=white, very thin] (f) circle (.25);
\draw[fill=white, very thin] (h) circle (.25);
\draw[fill=white, very thin] (j) circle (.25);
\draw[fill=white, very thin] (k) circle (.25);
\draw[fill=white, very thin] (m) circle (.25);
\draw[fill=white, very thin] (o) circle (.25);
\path[fill=white] (p) circle (.45);
\path[fill=white] (q) circle (.35);
\path[fill=white] (r) circle (.35);
\node at (a) {$#1$};
\node at (c) {$#2$};
\node at (e) {$#3$};
\node at (f) {$#4$};
\node at (h) {$#5$};
\node at (j) {$#6$};
\node at (k) {$#7$};
\node at (m) {$#8$};
\node at (o) {$#9$};
\node at (p) {$R^{X Y}_{z_1,z_2}$};
\node at (q) {$R^{X}_{z_1}$};
\node at (r) {$R^{Y}_{z_2}$};
\end{tikzpicture}}}}

\newcommand{\rhs}[9]{\raisebox{-49pt}{
\resizebox{5.4cm}{4.5cm}
{\begin{tikzpicture}
\path[use as bounding box](-1,-2) rectangle(5,3);
\coordinate (a) at (0,-1); %1
\coordinate (b) at (0,1);  %2
\coordinate (c) at (1,2);  %3
\coordinate (d) at (4,1);  %4
\coordinate (e) at (4,-1); %5
\coordinate (f) at (1,-2); %6
\coordinate (g) at (2,-1); %7
\coordinate (h) at (2,1);  %8
\coordinate (i) at (1,0);  %9
\coordinate (r) at (3,0);  %R
\coordinate (s) at (1,-1);  %R
\coordinate (t) at (1,1);  %R
\draw[semithick] (a)--(g) to [out=0, in=180] (d);
\draw[semithick] (b)--(h) to [out=0, in=180] (e);
\draw[semithick] (c)--(f);
\draw[fill=white, very thin] (a) circle (.25);
\draw[fill=white, very thin] (b) circle (.25);
\draw[fill=white, very thin] (c) circle (.25);
\draw[fill=white, very thin] (d) circle (.25);
\draw[fill=white, very thin] (e) circle (.25);
\draw[fill=white, very thin] (f) circle (.25);
\draw[fill=white, very thin] (g) circle (.25);
\draw[fill=white, very thin] (h) circle (.25);
\draw[fill=white, very thin] (i) circle (.25);
\path[fill=white] (r) circle (.45);
\path[fill=white] (s) circle (.35);
\path[fill=white] (t) circle (.35);
\node at (a) {$#1$};
\node at (b) {$#2$};
\node at (c) {$#3$};
\node at (d) {$#4$};
\node at (e) {$#5$};
\node at (f) {$#6$};
\node at (g) {$#7$};
\node at (h) {$#8$};
\node at (i) {$#9$};
\node at (r) {$R^{X Y}_{z_1, z_2}$};
\node at (s) {$R^X_{z_1}$};
\node at (t) {$R^{Y}_{z_2}$};
\end{tikzpicture}}}}

\newcommand{\lhsgd}[9]{\raisebox{-49pt}{\resizebox{5.4cm}{4.5cm}{\begin{tikzpicture}
\path[use as bounding box](-1,-2) rectangle(5,3);
\coordinate (a) at (0,-1);     % 1
\coordinate (b) at (0,-1.5);   % R
\coordinate (c) at (0,1);      % 2
\coordinate (d) at (0,1.5);    % 
\coordinate (e) at (3,2);      % 3
\coordinate (f) at (4,1);      % 4
\coordinate (g) at (4,1.5);    % 
\coordinate (h) at (4,-1);     % 5
\coordinate (i) at (4,-1.5);   % 
\coordinate (j) at (3,-2);     % 6
\coordinate (k) at (2,1);      % 7
\coordinate (l) at (2,1.5);    % 
\coordinate (m) at (2,-1);     % 8
\coordinate (n) at (2,-1.5);   % 
\coordinate (o) at (3,0);      % 9
\coordinate (p) at (1,0);      % T
\coordinate (q) at (3,1);      % 
\coordinate (r) at (3,-1);     % 
\draw[semithick] (a) to [out=0, in=180] (k);
\draw[semithick] (k)--(f);
\draw[semithick] (c) to [out=0, in=180] (m);
\draw[semithick] (m)--(h);
\draw[semithick] (e)--(j);
\draw[fill=white, very thin] (a) circle (.25);
\draw[fill=white, very thin] (c) circle (.25);
\draw[fill=white, very thin] (e) circle (.25);
\draw[fill=white, very thin] (f) circle (.25);
\draw[fill=white, very thin] (h) circle (.25);
\draw[fill=white, very thin] (j) circle (.25);
\draw[fill=white, very thin] (k) circle (.25);
\draw[fill=white, very thin] (m) circle (.25);
\draw[fill=white, very thin] (o) circle (.25);
\path[fill=white] (p) circle (.45);
\path[fill=white] (q) circle (.35);
\path[fill=white] (r) circle (.35);
\node at (a) {$#1$};
\node at (c) {$#2$};
\node at (e) {$#3$};
\node at (f) {$#4$};
\node at (h) {$#5$};
\node at (j) {$#6$};
\node at (k) {$#7$};
\node at (m) {$#8$};
\node at (o) {$#9$};
\node at (p) {$R^{\Gamma \Delta}_{z_1,z_2}$};
\node at (q) {$R^{\Gamma}_{z_1}$};
\node at (r) {$R^{\Delta}_{z_2}$};
\end{tikzpicture}}}}

\newcommand{\rhsgd}[9]{\raisebox{-49pt}{\resizebox{5.4cm}{4.5cm}
{\begin{tikzpicture}
\path[use as bounding box](-1,-2) rectangle(5,3);
\coordinate (a) at (0,-1); %1
\coordinate (b) at (0,1);  %2
\coordinate (c) at (1,2);  %3
\coordinate (d) at (4,1);  %4
\coordinate (e) at (4,-1); %5
\coordinate (f) at (1,-2); %6
\coordinate (g) at (2,-1); %7
\coordinate (h) at (2,1);  %8
\coordinate (i) at (1,0);  %9
\coordinate (r) at (3,0);  %R
\coordinate (s) at (1,-1);  %R
\coordinate (t) at (1,1);  %R
\draw[semithick] (a)--(g) to [out=0, in=180] (d);
\draw[semithick] (b)--(h) to [out=0, in=180] (e);
\draw[semithick] (c)--(f);
\draw[fill=white, very thin] (a) circle (.25);
\draw[fill=white, very thin] (b) circle (.25);
\draw[fill=white, very thin] (c) circle (.25);
\draw[fill=white, very thin] (d) circle (.25);
\draw[fill=white, very thin] (e) circle (.25);
\draw[fill=white, very thin] (f) circle (.25);
\draw[fill=white, very thin] (g) circle (.25);
\draw[fill=white, very thin] (h) circle (.25);
\draw[fill=white, very thin] (i) circle (.25);
\path[fill=white] (r) circle (.45);
\path[fill=white] (s) circle (.35);
\path[fill=white] (t) circle (.35);
\node at (a) {$#1$};
\node at (b) {$#2$};
\node at (c) {$#3$};
\node at (d) {$#4$};
\node at (e) {$#5$};
\node at (f) {$#6$};
\node at (g) {$#7$};
\node at (h) {$#8$};
\node at (i) {$#9$};
\node at (r) {$R^{\Gamma \Delta}_{z_1, z_2}$};
\node at (s) {$R^\Gamma_{z_1}$};
\node at (t) {$R^{\Delta}_{z_2}$};
\end{tikzpicture}}}}

\newcommand{\botcharge}[2]{\raisebox{-29pt}{$\begin{array}{cc}#1\\\scriptstyle #2\end{array}$}}
\newcommand{\topcharge}[2]{\raisebox{23pt}{$\begin{array}{cc}\scriptstyle #2\\#1\end{array}$}}

\newcommand{\XYtable}[9]{%
\begin{tikzpicture}[scale=0.5, every node/.style={scale=0.8, inner sep=0, outer sep=0}, baseline=-3.5]
\coordinate (a) at (-.75, -.75);
\coordinate (b) at (-.75, .75);
\coordinate (c) at (.75, .75);
\coordinate (d) at (.75, -.75);
\coordinate (aa) at (-.75, -1.32);
\coordinate (bb) at (-.75, 1.32);
\coordinate (cc) at (.75, 1.32);
\coordinate (dd) at (.75, -1.32);
\coordinate (ee) at (0,0);
\draw[semithick] (a) to [out=0, in=180] (c);
\draw[semithick] (b) to [out=0, in=180] (d);
\draw[fill=white, very thin] (a) circle (.3);
\draw[fill=white, very thin] (b) circle (.3);
\draw[fill=white, very thin] (c) circle (.3);
\draw[fill=white, very thin] (d) circle (.3);
\path[fill=white] (ee) circle (.52);
\node at (a) {$#1$};%
\node at (b) {$#2$};%
\node at (c) {$#3$};%
\node at (d) {$#4$};%
\node at (aa) {$#5$};%
\node at (bb) {$#6$};%
\node at (cc) {$#7$};%
\node at (dd) {$#8$};%
\node at (ee) {$#9$};%
\end{tikzpicture}%}
}

\theoremstyle{definition}
\newtheorem*{dbc}{Lattice and Boundary Spins Assignment}

\begin{document}

\title[Duality for Metaplectic Ice]{Duality for Metaplectic Ice}
\author{Ben Brubaker}
\address{School of Mathematics, University of Minnesota, Minneapolis, MN 55455}
\email{brubaker@math.umn.edu}
\author{Valentin Buciumas}
\address{Einstein Institute of Mathematics, Edmond J. Safra Campus, Givat Ram, The Hebrew University of Jerusalem, Jerusalem, 91904, Israel}
\email{valentin.buciumas@gmail.com}
\author{Daniel Bump}
\address{Department of Mathematics, Stanford University, Stanford, CA 94305-2125}
\email{bump@math.stanford.edu}
\author{Nathan Gray}
\address{Department of Mathematics, Mount Holyoke College, South Hadley, MA 01075}
\email{ngray@mtholyoke.edu}

\begin{abstract}
We interpret values of spherical Whittaker functions on metaplectic covers of the
general linear group over a nonarchimedean local field as partition functions of 
two different solvable lattice models. We prove the equality of these two 
partition functions by showing the commutativity of transfer matrices associated to different models via the Yang-Baxter equation.
\end{abstract}
\maketitle

\section{Introduction}

Let $F$ be a nonarchimedean local field containing sufficiently many $2 n$-th
roots of unity, and let $\GL (r, F)^{(n)}$ refer to an $n$-fold
metaplectic cover. We will use the term ``metaplectic ice'' to refer to
explicit realizations of spherical Whittaker functions on these groups as
partition functions of solvable
lattice models. See~{\cite{wmd5book}}, {\cite{mice}}, {\cite{BBB}} for
more information about these models. The purpose of this paper is to prove 
that two such realizations are equivalent.

To explain the issue let us consider the simplest case
$n = 1$, so that the statement is really about spherical Whittaker functions
for the reductive group $\GL (r, F)$.
If $\lambda = (\lambda_1, \ldots, \lambda_r)$ is an
element of the coweight lattice $\Z^r$, we say $\lambda$ is
\textit{dominant} if $\lambda_1 \geqslant \cdots \geqslant \lambda_r$. Let $\varpi$
be a prime element of $F$, let $q$ be the residue field cardinality,
and let $\varpi^{\lambda}$ denote $\mathrm{diag}
(\varpi^{\lambda_1}, \ldots, \varpi^{\lambda_r})$.

In this $n = 1$ case, Whittaker models of principal series representations are
unique, but the spherical Whittaker function admits multiple explicit inductive descriptions 
according to the order in which unipotent integrations are performed (corresponding to a chain
of parabolic subgroups from a chosen Borel subgroup up to $G$). Let us fix an
unramified principal series representation whose character
is given by Langlands parameters $\mathbf{z} = (z_1, \ldots, z_r)$. Two different means of evaluation produce, to each dominant coweight $\lambda$, a pair of formulas for the spherical Whittaker function:
\begin{equation} W_\Gamma(\varpi^\lambda) = \sum_{v \in
  \mathcal{B}_{\lambda+\rho}} H_\Gamma(v; \mathbf{z}), %\quad \text{and}
  \qquad 
  W_\Delta(\varpi^\lambda) = \sum_{v \in \mathcal{B}_{\lambda+\rho}} H_\Delta(v;
  \mathbf{z}).
\label{gammadeltawhittaker}
\end{equation}
Here $\lambda$ is a partition of length $\leqslant r$, regarded as a dominant weight, 
$\rho$ is the the Weyl vector, and $\mathcal{B}_{\lambda+\rho}$ is
the crystal with highest weight $\lambda+\rho$.
Since both are formulas for the spherical Whittaker function at the same point
for the same principal series, they are of course equal.  Furthermore,
according to the Shintani-Casselman-Shalika formula, both sums are thus equal
to:
\begin{equation}
  \label{whittakerschur} W (\varpi^{\lambda}) = \prod_{\alpha \in \Phi^+}
  (1 - q^{- 1} \mathbf{z}^{\alpha}) s_{\lambda} (z_1, \ldots, z_r)
\end{equation}
where $s_{\lambda}$ is the Schur polynomial.

Tokuyama~\cite{Tokuyama} proved that the function $\sum_{v} H_\Gamma(v;
\mathbf{z})$ is equal to $(1 - q^{- 1} \mathbf{z}^{\alpha}) s_{\lambda}
(z_1, \ldots, z_r)$, and a similar argument can be given for $\sum_{v}
H_\Delta(v; \mathbf{z})$. But one may then dispense with the intermediate step involving the
Whittaker function and ask for a direct proof that
$W_\Gamma(\varpi^\lambda)=W_\Delta(\varpi^\lambda)$. This
is a special case of the problem we will solve in this paper.

A \textit{lattice model} is a statistical-mechanical model on a (for us, two-dimensional) lattice such as those
studied in {\cite{Baxter}}. \textit{Boltzmann weights} are attached to each state as a product of local
Boltzmann weights at each vertex, and then the \textit{partition function} 
is the sum of the Boltzmann weights of all such states.
The model is termed \textit{solvable} if
it is amenable to study by means of the Yang-Baxter equation, as in {\cite{Baxter}}, Chapter~10. The term is apt
because one may often solve explicitly for the value of the partition function in such cases. 

It was realized by Hamel and King~\cite{HamelKing} that Tokuyama's formula
can be rewritten as the partition function of an ``ice-type'' solvable lattice
model. The states of the model are in bijection with a subset
of $\mathcal{B}_{\lambda+\rho}$. This rewriting of $H_\Gamma$ is
elementary, but significant since it makes available a different
tool, the Yang-Baxter equation. In \cite{hkice}, it was shown that both functions in
(\ref{gammadeltawhittaker}) are representable in terms of
lattice models and that Yang-Baxter equations gave both new
proofs of Tokuyama's formula and of the equality of the two functions.
Demonstrating this latter equality is difficult, as no bijection matching
summands of the $W_\Gamma$ and $W_\Delta$ exists (see {\cite{wmd5book}}
for details).

In the present paper ($n \geqslant 1$), the solvability of the model is reflected in functional equations for the partition function under the action
of the Weyl group on the spectral (Langlands) parameters $\mathbf{z}$. For the $n = 1$
case, this is equivalent to the permutation symmetry of the Schur polynomial $s_{\lambda}$.

For general $n$, one can still write values of Whittaker functions on the metaplectic $n$-cover of $\operatorname{GL}(r,F)$ as partition functions of two models with different Boltzmann weights \cite{mice}. 
We refer to them as \textit{Gamma ice} and 
\textit{Delta ice} in accordance with the labels in (\ref{gammadeltawhittaker}). Precise definitions of the models and their weights
are given in subsequent sections, but the above is enough to follow the present discussion. As we will explain, the use of
Yang-Baxter equations allows us to efficiently solve the following two problems for metaplectic covers of the general linear group.

\begin{problem}
  \label{probone}Prove, for any choice of dominant weight $\lambda$, functional equations for this pair of partition functions corresponding to the Weyl
  group action on the spectral (Langlands) parameters~$z_i$.
\end{problem}

\begin{problem}
  \label{probtwo} Prove, for any choice of dominant weight $\lambda$, that the two partition functions, ``Gamma ice'' and ``Delta ice'', are equal 
  (without using the fact that the models represent metaplectic Whittaker functions).
\end{problem}

If $n = 1$, both Problems~\ref{probone} and~\ref{probtwo} have
been solved using the Yang-Baxter equation. (See Chapter~19 of
{\cite{wmd5book}}, or the arxiv version of {\cite{hkice}}.)
It is natural to expect this for general $n$, since both
problems may be expressed in terms of the commutation
properties of {\textit{row transfer matrices}} for the models,
and as in {\cite{Baxter}}, Chapter 10, the Yang-Baxter equation
is a natural tool for such commutativity results. So the two problems described above were set out in {\cite{mice}}, where
lattice models were given representing a spherical metaplectic
Whittaker function, generalizing Gamma and Delta ice to $n>1$.
It was proposed that the Yang-Baxter equation could be used
this way. In retrospect we understand that the lattice models set out in
{\cite{mice}} were not solvable, even though they gave the right partition
functions. So (at that time) it was not possible to carry out the proof.

In {\cite{BBB}}, the first three authors found modifications of these models
that give the same partition function (for combinatorial reasons) and these
models \textit{are} solvable. The Yang-Baxter equations required to solve
Problem~\ref{probone} involve the $R$-matrix for modules of a certain quantum group, a
Drinfeld twist of the quantum affine Lie superalgebra $U_{q^{-1/2}}
(\widehat{\mathfrak{gl}} (n| 1))$. If $n = 1$, these models are the same as
those in {\cite{mice}}, otherwise they are slightly different, and the
modification is essential for the Yang-Baxter equation. See also \cite{BBBF},
where the Drinfeld twisting is further explained, and applications to
representations of Hecke algebras are given.

Problem~\ref{probone} is solved in~\cite{BBB} for Gamma ice. Even though Delta ice is not mentioned in that paper, the solution for Delta ice is similar. 
Another way to think about Problem~\ref{probone} is in the setting of Chinta and Gunnells {\cite{ChintaGunnells}}. Their results may be understood as giving a solution
to Problem~\ref{probone}, though not involving the Yang-Baxter equation.
Considerable subtlety is involved in their definition of an action of the Weyl
group on spectral parameters, related to the scattering matrix of the
intertwining operators for principal series on the metaplectic group, computed
earlier in~{\cite{KazhdanPatterson}}.

Regarding Problem~\ref{probtwo}, in \cite{wmd5book}, the equivalence of Gamma
and Delta ice was used to prove the analytic continuation and functional
equations of certain multiple Dirichlet series, but its significance goes
beyond this application. The phenomenon that two systems can
have the same (or closely related) partition functions occurs frequently in physics.
For example Kramers-Wannier~\cite{KramersWannier} duality relates the partition functions
of the low- and high-temperature Ising models. Or in the dualities
of string theory, different theories may represent the same universe,
and this is reflected in the equality of their partition functions.
See for example~\cite{BeckerBeckerSchwartz}, or for a specific
case \cite{GirardelloEtAl} (6.31). We believe the relationship
between Gamma and Delta ice is analogous to such dualities.

A solution to Problem~2 was given in {\cite{wmd5book}}. This gave a direct
connection between the representations of metaplectic Whittaker functions by
Gamma ice and Delta ice. Again, this proof did not depend on the Yang-Baxter
equation. The partition function is equal to a sum over Gelfand-Tsetlin
patterns, and after a first reduction of the problem, a combinatorial
statement was given that would solve Problem~\ref{probtwo}. This fact, called
\textit{Statement B}, is equivalent to the commutativity of two transfer
matrices, associated with Gamma ice and Delta ice.

Now the proof of Statement~B in {\cite{wmd5book}} was extremely intricate, and
far more complex than a proof using the Yang-Baxter equation would be. We may
use the language of crystals to describe it. The Gamma and Delta ice
descriptions can be translated into sums over crystal bases, and the
Sch{\"u}tzenberger involution of the crystal roughly matches terms of
these two sums. The crystal may be embedded as the set of lattice points in a
polytope. For elements in the interior of the polytope, terms corresponding by
the Sch{\"u}tzenberger involution are exactly equal. But unfortunately
the contributions of terms on the boundary do not exactly match; rather
different terms must be grouped together into small ``packets'' in order to
see that the two sums are equal. In complicated cases the packets become
chaotically larger and are not simple to describe and compare. Thus in a
particular case the required identity can be checked, but the combinatorics of
a complete proof seem very perplexing. The difficulties could ultimately be
overcome, and although the details are fascinating, the proof is very long.
Later \cite{FriedbergZhang} gave another application of these techniques.

Thus a proof of Statement~B using the Yang-Baxter equation is highly
desirable, but was not available when {\cite{wmd5book}} was written.
Now it is possible to give such a proof. In this paper, we give a proof of
commutativity of transfer matrices for Gamma and Delta ice which is equivalent
to Statement~B (Theorem \ref{theorem:Bprime}). The proof is based on new
Yang-Baxter equations, relying on the results of versions of the Yang-Baxter
equation from {\cite{BBB}} and {\cite{Graythesis}}. We then show how to use
this result to solve Problem~2 (Theorem \ref{theorem:Aprime}).

There is another context in which Gamma and Delta ice occur. In
\cite{wmd5book}, the Gamma and Delta ice models relate two
systems, one made entirely of Gamma ice, the other of Delta
ice. But Gamma and Delta ice also appear together in models for
$p$-adic Whittaker functions of Type C. In these models,
layers of Gamma and Delta ice are arranged in alternating
rows. See {\cite{Ivanov, BBCG}} for the $n=1$ case, where such models are
applied, respectively to Whittaker functions on odd orthogonal groups, and on
the metaplectic double covers of $\Sp (2 r, F)$, where $F$ is a nonarchimedean
local field. This work was generalized in~\cite{Graythesis} to general~$n$.
Several $R$-matrices and Yang-Baxter equations are needed to solve
Problem~\ref{probtwo}, and these are the same ones that are
needed for these problems with Type~C Whittaker functions.

The fact that both Delta and Gamma ice appear simultaneously in the settings
mentioned above suggests there may be a quantum group whose modules admit
endomorphisms using both flavors of ice. More explicitly, we seek a
quasitriangular Hopf algebra $H$ with a universal $R$-matrix
$\mathcal{R} \in H \otimes H$ and two parametrized families of modules
$V_{\Delta}$ and $V_{\Gamma}$ such that the action of $\mathcal{R}$ on
$V_{X} \otimes V_{Y}$ equals $R^{X Y}$ for $X, Y$ any combination of
$\Gamma$ and $\Delta$. In Section \ref{section:YangBaxterSystem} we prove that
the matrices $R^{\Gamma \Gamma}, (R^{\Delta \Gamma})^{-1}, R^{\Gamma \Delta}$
and $(R^{\Delta \Delta})^\ddagger$ form a Yang-Baxter system in the sense of
\cite{HlavatySnobl,HlavatyAlgebraic,HlavatyApplications}, where the involution
$\ddagger$ is defined in Section~\ref{section:YangBaxterSystem}. Moreover, the
endomorphisms $(R^{\Gamma \Delta})^\ddagger$ and $(R^{\Delta \Gamma})^{-1}$
agree up to a scalar multiple. Using a theorem of Freidel and
Maillet~\cite{FreidelMaillet}, these statements imply the existence of such a
quantum group via a parametrized version of the FRT construction.
(More precisely they produce a dual quasitriangular Hopf algebra,
with $V_\Delta$ and $V_\Gamma$ as comodules.) This observation is preliminary and so
will say nothing about its representation theory but perhaps it may be thought
of as a kind of quantum double.

\textbf{Acknowledgements:} This work was supported by NSF grants DMS-1406238 (Brubaker), 
DMS-1601026 (Bump) and by the Max Planck Institute for Mathematics (Buciumas). 

\section{Metaplectic ice}

The models for metaplectic ice are six-vertex models
with boundary conditions generalizing the so-called ``domain wall boundary''
as used in \cite{Izergin,KuperbergASM}.
Each edge in the model is assigned a {\it spin} ``$+$'' or ``$-$'' according to the familiar
ice rule resulting in six allowable configurations; these appear in
Figure~\ref{mweights}. We begin by assigning a rectangular lattice and
its boundary spins according to a partition $\lambda$ with
$r$ parts (some of which may be 0) as follows:

\begin{dbc} Given such a partition $\lambda = (\lambda_1, \lambda_2, \ldots, \lambda_r)$, we form a lattice with $r$ rows and $M := \lambda_1+r$ columns and define boundary conditions of a system as follows:
\begin{itemize}
\item the left and bottom edges have spin $+$;
\item the right edges have spin $-$;
\item the top edges in columns $\lambda_1+r-1$, $\lambda_2+r-2$, \ldots, $\lambda_r$ have spin $-$ while the rest of the top edges have spin $+$.  
\end{itemize}
\end{dbc}

An \textit{admissible state} of the model is then an assignment of spins to the interior edges such that each vertex has adjacent edges in one of the six allowable configurations; we
will present two types of ice, $\Gamma$ and $\Delta$, whose allowable configurations of spins are presented in Figure~\ref{mweights}.
Figure~\ref{fig:mice} gives an example of an admissible state for a ($\Gamma$-ice) model whose boundary conditions are dictated by the partition $\lambda = (3,2,0)$.

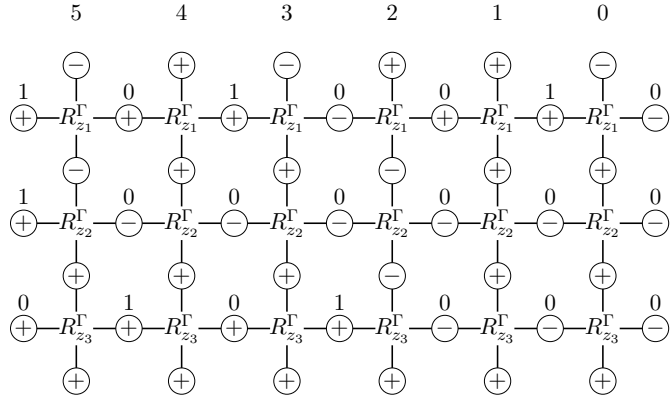
\begin{figure}[h]
%\scalebox{0.8}{
\begin{tikzpicture}[scale=0.7, every node/.style={scale=0.8}]
  \coordinate (ab) at (1,0);
  \coordinate (ad) at (3,0);
  \coordinate (af) at (5,0);
  \coordinate (ah) at (7,0);
  \coordinate (aj) at (9,0);
  \coordinate (al) at (11,0);
  \coordinate (ba) at (0,1);
  \coordinate (bc) at (2,1);
  \coordinate (be) at (4,1);
  \coordinate (bg) at (6,1);
  \coordinate (bi) at (8,1);
  \coordinate (bk) at (10,1);
  \coordinate (bm) at (12,1);
  \coordinate (baUP) at (0,1.5);
  \coordinate (bcUP) at (2,1.5);
  \coordinate (beUP) at (4,1.5);
  \coordinate (bgUP) at (6,1.5);
  \coordinate (biUP) at (8,1.5);
  \coordinate (bkUP) at (10,1.5);
  \coordinate (bmUP) at (12,1.5);

  \coordinate (cb) at (1,2);
  \coordinate (cd) at (3,2);
  \coordinate (cf) at (5,2);
  \coordinate (ch) at (7,2);
  \coordinate (cj) at (9,2);
  \coordinate (cl) at (11,2);
  \coordinate (da) at (0,3);
  \coordinate (dc) at (2,3);
  \coordinate (de) at (4,3);
  \coordinate (dg) at (6,3);
  \coordinate (di) at (8,3);
  \coordinate (dk) at (10,3);
  \coordinate (dm) at (12,3);
  \coordinate (daUP) at (0,3.5);
  \coordinate (dcUP) at (2,3.5);
  \coordinate (deUP) at (4,3.5);
  \coordinate (dgUP) at (6,3.5);
  \coordinate (diUP) at (8,3.5);
  \coordinate (dkUP) at (10,3.5);
  \coordinate (dmUP) at (12,3.5);

  \coordinate (eb) at (1,4);
  \coordinate (ed) at (3,4);
  \coordinate (ef) at (5,4);
  \coordinate (eh) at (7,4);
  \coordinate (ej) at (9,4);
  \coordinate (el) at (11,4);
  \coordinate (fa) at (0,5);
  \coordinate (fc) at (2,5);
  \coordinate (fe) at (4,5);
  \coordinate (fg) at (6,5);
  \coordinate (fi) at (8,5);
  \coordinate (fk) at (10,5);
  \coordinate (fm) at (12,5);
  \coordinate (faUP) at (0,5.5);
  \coordinate (fcUP) at (2,5.5);
  \coordinate (feUP) at (4,5.5);
  \coordinate (fgUP) at (6,5.5);
  \coordinate (fiUP) at (8,5.5);
  \coordinate (fkUP) at (10,5.5);
  \coordinate (fmUP) at (12,5.5);

  \coordinate (gb) at (1,6);
  \coordinate (gd) at (3,6);
  \coordinate (gf) at (5,6);
  \coordinate (gh) at (7,6);
  \coordinate (gj) at (9,6);
  \coordinate (gl) at (11,6);
  \coordinate (bb) at (1,1);
  \coordinate (bd) at (3,1);
  \coordinate (bf) at (5,1);
  \coordinate (bh) at (7,1);
  \coordinate (bj) at (9,1);
  \coordinate (bl) at (11,1);
  \coordinate (db) at (1,3);
  \coordinate (dd) at (3,3);
  \coordinate (df) at (5,3);
  \coordinate (dh) at (7,3);
  \coordinate (dj) at (9,3);
  \coordinate (dl) at (11,3);
  \coordinate (fb) at (1,5);
  \coordinate (fd) at (3,5);
  \coordinate (ff) at (5,5);
  \coordinate (fh) at (7,5);
  \coordinate (fj) at (9,5);
  \coordinate (fl) at (11,5);
  \draw[semithick] (ab)--(gb);
  \draw[semithick] (ad)--(gd);
  \draw[semithick] (af)--(gf);
  \draw[semithick] (ah)--(gh);
  \draw[semithick] (aj)--(gj);
  \draw[semithick] (al)--(gl);
  \draw[semithick] (ba)--(bm);
  \draw[semithick] (da)--(dm);
  \draw[semithick] (fa)--(fm);
  \foreach \P in {(ab), (ad), (af), (ah), (aj), (al), (ba), (bc), (be), (bg), (bi), (bk), (bm), (cb), (cd), (cf), (ch), (cj), (cl), (da), (dc), (de), (dg), (di), (dk), (dm), (eb), (ed), (ef), (eh), (ej), (el), (fa), (fc), (fe), (fg), (fi), (fk), (fm), (gb), (gd), (gf), (gh), (gj), (gl)}
  {%
  \draw[fill=white, very thin] \P circle (.25);
  }%
  \foreach \P in {(bb), (bd), (bf), (bh), (bj), (bl), (db), (dd), (df), (dh), (dj), (dl), (fb), (fd), (ff), (fh), (fj), (fl)}
  {%
  \path[fill=white] \P circle (.32);
  }%
  \node at (bb) {$R^\Gamma_{z_3}$};
  \node at (bd) {$R^\Gamma_{z_3}$};
  \node at (bf) {$R^\Gamma_{z_3}$};
  \node at (bh) {$R^\Gamma_{z_3}$};
  \node at (bj) {$R^\Gamma_{z_3}$};
  \node at (bl) {$R^\Gamma_{z_3}$};
  \node at (db) {$R^\Gamma_{z_2}$};
  \node at (dd) {$R^\Gamma_{z_2}$};
  \node at (df) {$R^\Gamma_{z_2}$};
  \node at (dh) {$R^\Gamma_{z_2}$};
  \node at (dj) {$R^\Gamma_{z_2}$};
  \node at (dl) {$R^\Gamma_{z_2}$};
  \node at (fb) {$R^\Gamma_{z_1}$};
  \node at (fd) {$R^\Gamma_{z_1}$};
  \node at (ff) {$R^\Gamma_{z_1}$};
  \node at (fh) {$R^\Gamma_{z_1}$};
  \node at (fj) {$R^\Gamma_{z_1}$};
  \node at (fl) {$R^\Gamma_{z_1}$};
  \node at (gb) {$-$};
  \node at (gd) {$+$};
  \node at (gf) {$-$};
  \node at (gh) {$+$};
  \node at (gj) {$+$};
  \node at (gl) {$-$};
  \node at (fa) {$+$};
  \node at (fc) {$+$};
  \node at (fe) {$+$};
  \node at (fg) {$-$};
  \node at (fi) {$+$};
  \node at (fk) {$+$};
  \node at (fm) {$-$};
  \node at (eb) {$-$};
  \node at (ed) {$+$};
  \node at (ef) {$+$};
  \node at (eh) {$-$};
  \node at (ej) {$+$};
  \node at (el) {$+$};
  \node at (da) {$+$};
  \node at (dc) {$-$};
  \node at (de) {$-$};
  \node at (dg) {$-$};
  \node at (di) {$-$};
  \node at (dk) {$-$};
  \node at (dm) {$-$};
  \node at (cb) {$+$};
  \node at (cd) {$+$};
  \node at (cf) {$+$};
  \node at (ch) {$-$};
  \node at (cj) {$+$};
  \node at (cl) {$+$};
  \node at (ba) {$+$};
  \node at (bc) {$+$};
  \node at (be) {$+$};
  \node at (bg) {$+$};
  \node at (bi) {$-$};
  \node at (bk) {$-$};
  \node at (bm) {$-$};
  \node at (ab) {$+$};
  \node at (ad) {$+$};
  \node at (af) {$+$};
  \node at (ah) {$+$};
  \node at (aj) {$+$};
  \node at (al) {$+$};
  \node at (11,7) {$0$};
  \node at (9,7) {$1$};
  \node at (7,7) {$2$};
  \node at (5,7) {$3$};
  \node at (3,7) {$4$};
  \node at (1,7) {$5$};
  \node at (baUP) {$0$};
  \node at (bcUP) {$1$};
  \node at (beUP) {$0$};
  \node at (bgUP) {$1$};
  \node at (biUP) {$0$};
  \node at (bkUP) {$0$};
  \node at (bmUP) {$0$};
  \node at (daUP) {$1$};
  \node at (dcUP) {$0$};
  \node at (deUP) {$0$};
  \node at (dgUP) {$0$};
  \node at (diUP) {$0$};
  \node at (dkUP) {$0$};
  \node at (dmUP) {$0$};
  \node at (faUP) {$1$};
  \node at (fcUP) {$0$};
  \node at (feUP) {$1$};
  \node at (fgUP) {$0$};
  \node at (fiUP) {$0$};
  \node at (fkUP) {$1$};
  \node at (fmUP) {$0$};

\end{tikzpicture}%}
\caption{An admissible state of a metaplectic ice of type $\Gamma$ for $n=2$. 
The top boundary condition is given by the partition $\lambda = (3,2,0)$.  
}
\label{fig:mice}
\end{figure}

Each vertex in the model is assigned a Boltzmann weight from one of two sets
of weights called $R^\Gamma_{z_i}$ and $R^\Delta_{z_i}$, each depending on a
parameter $z_i \in \C^\times$, with $i = 1, \ldots, r$, according to the row
in which the vertex appears. We are ultimately interested in systems whose
vertices are all either type $\Delta$ or $\Gamma$, but in order to prove that
two such systems have equal partitions functions we will consider
``intermediate'' systems with a mix of rows of $\Delta$ vertices and rows of
$\Gamma$ vertices.\footnote{In~\cite{Graythesis} the systems of ultimate interest
are also ones in which rows of $\Gamma$- and $\Delta$-ice occur intermixed.}
In Figure~\ref{fig:mice} all vertices use a weight $R_{z_i}^\Gamma$ for $i=1,2,3$, and 
labels have been placed on the vertices to indicate the weighting function to be used.

In order to define these Boltzmann weights, we require one additional
statistic called \textit{charge}. This charge is an element of $\Z /n \Z $,
and is assigned to each horizontal edge in a state of ice.\footnote{Charge may
also be regarded as taking values in the non-negative integers, and all of our results will
remain true. However the Boltzmann weights that we use will only depend on the charges
modulo $n$. The case of values in $\Z /n \Z$ is most relevant to
metaplectic Whittaker functions on $n$-fold covers.}
In figures, we denote charge by representatives in $[0,n-1]$. To an admissible state of ice, charge
of an edge may be viewed as a global statistic. For rows of $\Gamma$-ice, it begins at 0 on
the right edge, and increments by 1 (mod $n$) at each edge with $+$ spin moving left-to-right.
Therefore, for rows of $\Gamma$-ice,
  the charge on a horizontal edge equals the number of $+$ edges
  (modulo $n$) in its row to the right of the given edge (including the
  edge in question). In Figure~\ref{fig:mice}, we have written the charge mod 2 above each horizontal edge.
For rows of $\Delta$-ice, we begin the charge at $0$ on the left edge and increment charge by 1 (mod $n$)
on each edge with a $-$ spin moving right-to-left. Therefore, for rows of $\Delta$-ice, the charge on a horizontal edge equals the number of $-$ edges (modulo $n$) in its row to the left of the given edge (including the edge in question). The possible values of charge for a given spin configuration are listed in Figure~\ref{mweights}.

\begin{remark}\label{remark:charge0}
As is clear from Figure~\ref{mweights}, the $\Gamma$-ice weights are 0 unless the horizontal edges with spin $-$ have associated charge $0$. Similarly for $\Delta$-ice weights, vertices with horizontal spin $+$ have weight 0 unless the associated charge is $0$. Thus each type of ice has $n$ admissible charges attached to one spin, and $1$ admissible charge associated to the other spin. This is the first hint of the connection to the defining module of a Drinfeld twist of $U_{q^{-1/2}} (\widehat{\mathfrak{gl}} (n| 1))$ of \cite{BBB} mentioned in the introduction. \end{remark}

\begin{remark} Charge may also be thought of as a local statistic. Indeed we could initially allow all possible choices of charge mod $n$ on horizontal edges, resulting in a $6n^2$-vertex model. Then taking weights as in Figure~\ref{mweights}, many of these vertices would receive Boltzmann weight 0, reducing to a $(2n+4)$-vertex model. If charge is regarded as a local statistic, then we also want to regard charges on boundary horizontal edges as part of the defining data of the system. Summing over all possible boundary charges gives the partition function we're studying at present.

\end{remark}

To each vertex we assign a $\textit{Boltzmann weight}$ according to Figure~\ref{mweights}, which depends on
the spin and charge of the edges incident to the vertex and on the type and
parameter of the vertex. The Boltzmann weights in Figure~\ref{mweights} 
depend on a parameter $v$ and a function $g(a)$ for a
charge $a$ mod $n$ that satisfies $g(0) = -v$ and $g(a) g(n-a) = v$ if $n$
does not divide $a$. In the connection with metaplectic Whittaker functions,
$v = q^{-1}$ where $q$ is the cardinality of the residue field of the
nonarchimedean local field $F$, and $g$ is a $p$-adic Gauss sum given by
integrating an additive and multiplicative character over the units in the
ring of integers of $F$. For a detailed definition of $g$ see for example
Section 3 in \cite{mice}, but note that for the existence of the Yang-Baxter
equation (and therefore for proving our theorems) we only use the two
properties of $g$ stated above. Any choice of adjacent spins and charges
that is not listed in Figure~\ref{mweights} has Boltzmann weight $0$.

The Boltzmann weight of a state is the product of the Boltzmann weights
of all vertices in the state.  The $\textit{partition function}$ of any such
model is the sum of the Boltzmann weights of each admissible state in the
model. Our goal is to
study identities among partition functions of the type described above.

\begin{figure}[h]
\[
\begin{array}{|c|c|c|c|c|c|c|}
\hline
&a_1&a_2&b_1&b_2&c_1&c_2\\
\hline
\Gamma \text{-ice} &
\begin{array}{c}\hspace{-6pt}\gammaice{+}{+}{+}{+}{a+1}{a}{R^\Gamma_z}\\1\end{array} &
\begin{array}{c}\gammaice{-}{-}{-}{-}{0}{0}{R^\Gamma_z}\\z\end{array} &
\begin{array}{c}\hspace{-6pt}\gammaice{+}{-}{+}{-}{a+1}{a}{R^\Gamma_z}\\g(a)\end{array} &
\begin{array}{c}\gammaice{-}{+}{-}{+}{0}{0}{R^\Gamma_z}\\z\end{array} &
\begin{array}{c}\gammaice{-}{+}{+}{-}{0}{0}{R^\Gamma_z}\\(1-v) z\end{array} &
\begin{array}{c}\gammaice{+}{-}{-}{+}{1}{0}{R^\Gamma_z}\\1\end{array} \\
\hline
\Delta \text{-ice} &
\begin{array}{c}\gammaice{+}{+}{+}{+}{0}{0}{R^\Delta_z}\\1\end{array} &
\begin{array}{c}\gammaice{-}{-}{-}{-}{a}{a+1}{R^\Delta_z}\hspace{-6pt}\\g(a) z\end{array} &
\begin{array}{c}\gammaice{+}{-}{+}{-}{0}{0}{R^\Delta_z}\\1\end{array} &
\begin{array}{c}\gammaice{-}{+}{-}{+}{a}{a+1}{R^\Delta_z}\hspace{-5pt}\\z\end{array} &
\begin{array}{c}\gammaice{-}{+}{+}{-}{0}{0}{R^\Delta_z}\\ (1-v) z\end{array} &
\begin{array}{c}\gammaice{+}{-}{-}{+}{0}{1}{R^\Delta_z}\\1\end{array} \\
\hline
\end{array}
\]
\caption{The Boltzmann weights for $\Gamma$ and $\Delta$ vertices associated to
a row parameter $z \in \mathbb{C}^\times$. The charge $a$ above an edge indicates any choice of
charge mod $n$ and gives the indicated weight. The weights depend on a parameter $v$ and any function $g$ with $g(0)=-v$
and $g(n-a)g(a) = v$ if $a \not\equiv 0$ mod $n$. If a configuration does not
appear in this table, its weight is zero.}
\label{mweights}
\end{figure}

For convenience, we refer to a rectangular-shaped model with boundary conditions and Boltzmann weights chosen as above as a \textit{system}. In particular, denote by $S_{\z, \lambda}^\Gamma$ (respectively $S^{\Delta}_{\z,\lambda}$) the system of metaplectic ice with all vertices of type $\Gamma$ (resp., $\Delta$), top row boundary conditions determined by $\lambda$ and for $\z = (z_1,\ldots, z_r)$, each vertex in row $i$ has parameter $z_i$.

Given a system $S$, let $Z(S)$ denote its partition function. If $\z = (z_1,\ldots, z_r)$, define $\z^\sigma := (z_r, z_{r-1}, \ldots, z_1)$. We can now state precisely the main theorem of this paper:

\begin{theorem}\label{theorem:Aprime}
The partition functions $Z(S^\Gamma_{\z,\lambda})$ and $Z(S^\Delta_{\z^\sigma\!,\lambda})$ are equal. 
\end{theorem} 
 
This is the solution to Problem 2 discussed in the introduction.
We will give a short proof of it in the next section using Theorem~\ref{theorem:Bprime}.

In order to reduce the theorem to a simpler statement, we define a two systems of two-row ice. Given a pair of partitions $\lambda = (\lambda_1, \ldots, \lambda_r)$ and $\mu = (\mu_1, \ldots, \mu_{r-1})$, and a pair $\z = (z_1, z_2) \in (\C^\times)^2$, we denote by $S_{\z, \lambda, \mu}^{\Gamma \Delta}$ the system of metaplectic ice with the following properties:
\begin{itemize}
\item the grid consists of $2$ rows and $M = \lambda_1+r$ columns;
\item the left and right boundary edges have spins $+$ and $-$, respectively; the top and bottom boundary edges have spins $-$ at the columns corresponding to parts of $\lambda$ and $\mu$, respectively;
\item the top row vertices have weights $R^\Gamma_{z_1}$ and the bottom row vertices have weights $R^\Delta_{z_2}$.
\end{itemize}
See Figure \ref{dgice} for an example of such a system $S^{\Gamma \Delta}_{\z, \lambda, \mu}$. We denote by $S^{\Delta \Gamma}_{\z^\sigma\!, \lambda, \mu}$ the system of metaplectic ice with the above conditions, but now with top row $\Delta$ and parameter $z_2$ and bottom row $\Gamma$ with parameter $z_1$. The boundary condition we impose is still charge $0$ for left boundary edges in a $\Delta$ row and right boundary edges in a $\Gamma$ row.

\begin{figure}[h]
%\scalebox{0.8}{
\begin{tikzpicture}[scale=0.7, every node/.style={scale=0.8}]  
  %\coordinate (ab) at (1,0);
  %\coordinate (ad) at (3,0);
  %\coordinate (af) at (5,0);
  %\coordinate (ah) at (7,0);
  %\coordinate (aj) at (9,0);
  %\coordinate (al) at (11,0);
  \coordinate (ba) at (0,1);
  \coordinate (bc) at (2,1);
  \coordinate (be) at (4,1);
  \coordinate (bg) at (6,1);
  \coordinate (bi) at (8,1);
  \coordinate (bk) at (10,1);
  \coordinate (bm) at (12,1);
  \coordinate (baUP) at (0,1.5);
  \coordinate (bcUP) at (2,1.5);
  \coordinate (beUP) at (4,1.5);
  \coordinate (bgUP) at (6,1.5);
  \coordinate (biUP) at (8,1.5);
  \coordinate (bkUP) at (10,1.5);
  \coordinate (bmUP) at (12,1.5);

  \coordinate (cb) at (1,2);
  \coordinate (cd) at (3,2);
  \coordinate (cf) at (5,2);
  \coordinate (ch) at (7,2);
  \coordinate (cj) at (9,2);
  \coordinate (cl) at (11,2);
  \coordinate (da) at (0,3);
  \coordinate (dc) at (2,3);
  \coordinate (de) at (4,3);
  \coordinate (dg) at (6,3);
  \coordinate (di) at (8,3);
  \coordinate (dk) at (10,3);
  \coordinate (dm) at (12,3);
  \coordinate (daUP) at (0,3.5);
  \coordinate (dcUP) at (2,3.5);
  \coordinate (deUP) at (4,3.5);
  \coordinate (dgUP) at (6,3.5);
  \coordinate (diUP) at (8,3.5);
  \coordinate (dkUP) at (10,3.5);
  \coordinate (dmUP) at (12,3.5);

  \coordinate (eb) at (1,4);
  \coordinate (ed) at (3,4);
  \coordinate (ef) at (5,4);
  \coordinate (eh) at (7,4);
  \coordinate (ej) at (9,4);
  \coordinate (el) at (11,4);
  \coordinate (fa) at (0,5);
  \coordinate (fc) at (2,5);
  \coordinate (fe) at (4,5);
  \coordinate (fg) at (6,5);
  \coordinate (fi) at (8,5);
  \coordinate (fk) at (10,5);
  \coordinate (fm) at (12,5);
  \coordinate (faUP) at (0,5.5);
  \coordinate (fcUP) at (2,5.5);
  \coordinate (feUP) at (4,5.5);
  \coordinate (fgUP) at (6,5.5);
  \coordinate (fiUP) at (8,5.5);
  \coordinate (fkUP) at (10,5.5);
  \coordinate (fmUP) at (12,5.5);

  \coordinate (gb) at (1,6);
  \coordinate (gd) at (3,6);
  \coordinate (gf) at (5,6);
  \coordinate (gh) at (7,6);
  \coordinate (gj) at (9,6);
  \coordinate (gl) at (11,6);
  %\coordinate (bb) at (1,1);
  %\coordinate (bd) at (3,1);
  %\coordinate (bf) at (5,1);
  %\coordinate (bh) at (7,1);
  %\coordinate (bj) at (9,1);
  %\coordinate (bl) at (11,1);
  \coordinate (db) at (1,3);
  \coordinate (dd) at (3,3);
  \coordinate (df) at (5,3);
  \coordinate (dh) at (7,3);
  \coordinate (dj) at (9,3);
  \coordinate (dl) at (11,3);
  \coordinate (fb) at (1,5);
  \coordinate (fd) at (3,5);
  \coordinate (ff) at (5,5);
  \coordinate (fh) at (7,5);
  \coordinate (fj) at (9,5);
  \coordinate (fl) at (11,5);
  \draw[semithick] (cb)--(gb);
  \draw[semithick] (cd)--(gd);
  \draw[semithick] (cf)--(gf);
  \draw[semithick] (ch)--(gh);
  \draw[semithick] (cj)--(gj);
  \draw[semithick] (cl)--(gl);
  \draw[semithick] (da)--(dm);
  \draw[semithick] (fa)--(fm);
  \draw[fill=white, very thin] (cb) circle (.25);
  \draw[fill=white, very thin] (cd) circle (.25);
  \draw[fill=white, very thin] (cf) circle (.25);
  \draw[fill=white, very thin] (ch) circle (.25);
  \draw[fill=white, very thin] (cj) circle (.25);
  \draw[fill=white, very thin] (cl) circle (.25);
  \draw[fill=white, very thin] (da) circle (.25);
  \draw[fill=white, very thin] (dc) circle (.25);
  \draw[fill=white, very thin] (de) circle (.25);
  \draw[fill=white, very thin] (dg) circle (.25);
  \draw[fill=white, very thin] (di) circle (.25);
  \draw[fill=white, very thin] (dk) circle (.25);
  \draw[fill=white, very thin] (dm) circle (.25);
  \draw[fill=white, very thin] (eb) circle (.25);
  \draw[fill=white, very thin] (ed) circle (.25);
  \draw[fill=white, very thin] (ef) circle (.25);
  \draw[fill=white, very thin] (eh) circle (.25);
  \draw[fill=white, very thin] (ej) circle (.25);
  \draw[fill=white, very thin] (el) circle (.25);
  \draw[fill=white, very thin] (fa) circle (.25);
  \draw[fill=white, very thin] (fc) circle (.25);
  \draw[fill=white, very thin] (fe) circle (.25);
  \draw[fill=white, very thin] (fg) circle (.25);
  \draw[fill=white, very thin] (fi) circle (.25);
  \draw[fill=white, very thin] (fk) circle (.25);
  \draw[fill=white, very thin] (fm) circle (.25);
  \draw[fill=white, very thin] (gb) circle (.25);
  \draw[fill=white, very thin] (gd) circle (.25);
  \draw[fill=white, very thin] (gf) circle (.25);
  \draw[fill=white, very thin] (gh) circle (.25);
  \draw[fill=white, very thin] (gj) circle (.25);
  \draw[fill=white, very thin] (gl) circle (.25);
  %
  %\path[fill=white] (bb) circle (.32);
  %\path[fill=white] (bd) circle (.32);
  %\path[fill=white] (bf) circle (.32);
  %\path[fill=white] (bh) circle (.32);
  %\path[fill=white] (bj) circle (.32);
  %\path[fill=white] (bl) circle (.32);
  \path[fill=white] (db) circle (.32);
  \path[fill=white] (dd) circle (.32);
  \path[fill=white] (df) circle (.32);
  \path[fill=white] (dh) circle (.32);
  \path[fill=white] (dj) circle (.32);
  \path[fill=white] (dl) circle (.32);
  \path[fill=white] (fb) circle (.32);
  \path[fill=white] (fd) circle (.32);
  \path[fill=white] (ff) circle (.32);
  \path[fill=white] (fh) circle (.32);
  \path[fill=white] (fj) circle (.32);
  \path[fill=white] (fl) circle (.32);
  \node at (db) {$R^\Delta_{z_2}$};
  \node at (dd) {$R^\Delta_{z_2}$};
  \node at (df) {$R^\Delta_{z_2}$};
  \node at (dh) {$R^\Delta_{z_2}$};
  \node at (dj) {$R^\Delta_{z_2}$};
  \node at (dl) {$R^\Delta_{z_2}$};
  \node at (fb) {$R^\Gamma_{z_1}$};
  \node at (fd) {$R^\Gamma_{z_1}$};
  \node at (ff) {$R^\Gamma_{z_1}$};
  \node at (fh) {$R^\Gamma_{z_1}$};
  \node at (fj) {$R^\Gamma_{z_1}$};
  \node at (fl) {$R^\Gamma_{z_1}$};
  \node at (gb) {$+$};
  \node at (gd) {$-$};
  \node at (gf) {$+$};
  \node at (gh) {$-$};
  \node at (gj) {$-$};
  \node at (gl) {$+$};
  \node at (fa) {$+$};
  \node at (fc) {};
  \node at (fe) {};
  \node at (fg) {};
  \node at (fi) {};
  \node at (fk) {};
  \node at (fm) {$-$};
  \node at (eb) {};
  \node at (ed) {};
  \node at (ef) {};
  \node at (eh) {};
  \node at (ej) {};
  \node at (el) {};
  \node at (da) {$+$};
  \node at (dc) {};
  \node at (de) {};
  \node at (dg) {};
  \node at (di) {};
  \node at (dk) {};
  \node at (dm) {$-$};
  \node at (cb) {$+$};
  \node at (cd) {$-$};
  \node at (cf) {$+$};
  \node at (ch) {$+$};
  \node at (cj) {$+$};
  \node at (cl) {$+$};
  \node at (11,7) {$0$};
  \node at (9,7) {$1$};
  \node at (7,7) {$2$};
  \node at (5,7) {$3$};
  \node at (3,7) {$4$};
  \node at (1,7) {$5$};
\end{tikzpicture}%}
\caption{$\smash{S^{\Gamma \Delta}_{\z, \lambda, \mu}}$-ice for $\lambda=(2,1,1)$ and $\mu = (4)$.  
}
\label{dgice}
\end{figure}

\begin{theorem}\label{theorem:Bprime}
Given partitions $\lambda$ and $\mu$ and any parameters $\z = (z_1, \ldots, z_r)$, $Z(S_{\z,\lambda,\mu}^{\Gamma \Delta}) = Z(S_{\z^\sigma\!,\lambda, \mu}^{\Delta \Gamma}).$ 
\end{theorem} 
 
We delay the proof until Section \ref{section:YBE} where we introduce our main tool, the Yang-Baxter equation. This theorem is equivalent to Statement B in \cite{wmd5book}; for details see Theorem 5 of \cite{mice}. We now prove Theorem \ref{theorem:Aprime} assuming Theorem \ref{theorem:Bprime}. 

\begin{proof}[Proof of Theorem \ref{theorem:Aprime}]
Let us start with $\smash{S^{\Gamma}_{\z,\lambda}}$-ice and show that its partition function is the same as that of $\smash{S^{\Delta}_{\z,\lambda}}$-ice in two steps. 

Let us consider, in addition to the system
$\smash{S^{\Gamma}_{\z,\lambda}}$, a system with identical boundary
conditions, except that all rows but the last are $\Gamma$-ice, but
the bottom row has been changed to $\Delta$-ice. We will call this
system $S'$. We will argue that there is a bijection between
the admissible states of $\smash{S^{\Gamma}_{\z,\lambda}}$ and those
of $S'$ in which corresponding states have the same Boltzmann
weight, so the two systems have the same partition function.

To see this recall that along the bottom boundary of $\Gamma$-ice, all the 
(vertical) edges have been assigned spin $+$. It is easy to see that in the
row above this, there must be exactly one vertical edge with spin $-$,
and all the other vertical edges will have spin $+$. Regarding the
horizontal edges in this bottom row, those to the left of the
unique $-$ vertical edge must have spin $+$, and those to the
right must have spin $-$. Thus the bottom row looks like this:
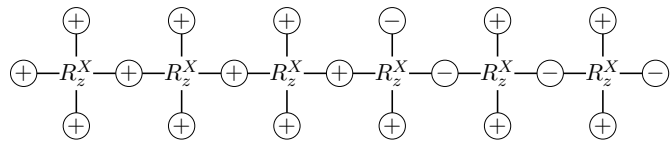
\[\begin{tikzpicture}[scale=0.7, every node/.style={scale=0.8}]
  \coordinate (ab) at (1,0);
  \coordinate (ad) at (3,0);
  \coordinate (af) at (5,0);
  \coordinate (ah) at (7,0);
  \coordinate (aj) at (9,0);
  \coordinate (al) at (11,0);
  \coordinate (cb) at (1,2);
  \coordinate (cd) at (3,2);
  \coordinate (cf) at (5,2);
  \coordinate (ch) at (7,2);
  \coordinate (cj) at (9,2);
  \coordinate (cl) at (11,2);
  \coordinate (ba) at (0,1);
  \coordinate (bb) at (1,1);
  \coordinate (bc) at (2,1);
  \coordinate (bd) at (3,1);
  \coordinate (be) at (4,1);
    \coordinate (bf) at (5,1);
  \coordinate (bg) at (6,1);
    \coordinate (bh) at (7,1);
  \coordinate (bi) at (8,1);
    \coordinate (bj) at (9,1);
  \coordinate (bk) at (10,1);
    \coordinate (bl) at (11,1);
  \coordinate (bm) at (12,1);
  \coordinate (baUP) at (0,1.5);
  \coordinate (bcUP) at (2,1.5);
  \coordinate (beUP) at (4,1.5);
  \coordinate (bgUP) at (6,1.5);
  \coordinate (biUP) at (8,1.5);
  \coordinate (bkUP) at (10,1.5);
  \coordinate (bmUP) at (12,1.5);
  \draw[semithick] (ab)--(cb);
  \draw[semithick] (ad)--(cd);
  \draw[semithick] (af)--(cf);
  \draw[semithick] (ah)--(ch);
  \draw[semithick] (aj)--(cj);
  \draw[semithick] (al)--(cl);
  \draw[semithick] (ba)--(bm);
  \foreach \P in {(ab), (ad), (af), (ah), (aj), (al), (ba), (bc), (be), (bg),
    (bi), (bk), (bm), (cb), (cd), (cf), (ch), (cj), (cl)}
  {%
  \draw[fill=white, very thin] \P circle (.25);
  }%
  \node at (ab) {$+$};
  \node at (ad) {$+$};
  \node at (af) {$+$};
  \node at (ah) {$+$};
  \node at (aj) {$+$};
  \node at (al) {$+$};
  \node at (cb) {$+$};
  \node at (cd) {$+$};
  \node at (cf) {$+$};
  \node at (ch) {$-$};
  \node at (cj) {$+$};
  \node at (cl) {$+$};
  \node at (ba) {$+$};
  \node at (bc) {$+$};
  \node at (be) {$+$};
  \node at (bg) {$+$};
  \node at (bi) {$-$};
  \node at (bk) {$-$};
  \node at (bm) {$-$};

  \path[fill=white] (bb) circle (.32);
  \node at (bb) {$R^X_{z}$};
  \path[fill=white] (bd) circle (.32);
  \node at (bd) {$R^X_{z}$};
  \path[fill=white] (bf) circle (.32);
  \node at (bf) {$R^X_{z}$};
  \path[fill=white] (bh) circle (.32);
  \node at (bh) {$R^X_{z}$};
  \path[fill=white] (bj) circle (.32);
  \node at (bj) {$R^X_{z}$};
  \path[fill=white] (bl) circle (.32);
  \node at (bl) {$R^X_{z}$};

\end{tikzpicture}%}
\]
Now, there is a unique way of assigning charges to the horizontal
edges here to an admissible state of either $\Gamma$-ice or $\Delta$-ice. Moreover, examining the Boltzmann weights of $\Gamma$-ice and $\Delta$-ice in Figure
\ref{mweights}, we notice that only the vertices of type $a_1$, $b_2$, and
$c_2$ have bottom edge with spin $+$, and for all these configurations the
weight of the vertex of type $\Gamma$ is the same as the weight of the
vertex of type $\Delta$. Moreover, choosing
the charges in the unique way that makes the state admissible,
the Boltzmann weight is $z_r^N$, where $N$ is the number of $b_2$
vertices in the row. The contribution to the partition functions
for $S'$ and $\smash{S^{\Gamma}_{\z,\lambda}}$ are the same,
and so these partition functions are equal.

We see that we may change the bottom row from $\Gamma$-ice to
$\Delta$-ice without changing the partition function. This would
not work for any other row.

The second step uses Theorem~\ref{theorem:Bprime} repeatedly. Since we have changed
the bottom row from $\Gamma$-ice to $\Delta$-ice, Theorem~\ref{theorem:Bprime} allows us
to interchange the bottom two rows and move the $\Delta$ row one step up. By doing this process repeatedly we move the $\Delta$ row with parameter $z_r$ to the top of the system without changing the partition function. 

Now the layer of $\Gamma$-ice with the $z_{r-1}$ parameter is at the bottom, and
as before we may change it to $\Delta$-ice without affecting the partition
function. Then we again move this row of $\Delta$-ice up using
Theorem~\ref{theorem:Bprime} until it reaches the second row. We continue the
process until all rows become $\Delta$-ice. The final model will be the system
$S^\Delta_{\z^\sigma\!,\lambda}$. 
\end{proof}

\section{The Yang-Baxter equation}\label{section:YBE}

In this section we prove Theorem \ref{theorem:Bprime} by a classical argument involving
the Yang-Baxter equation which was initially used by Baxter \cite{Baxter} to
solve the (classical, field-free) six- and eight-vertex models. In order to
explain our version of the Yang-Baxter equation, we first introduce another
family of vertices, which we'll refer to as $\textit{tilted}$ vertices,
whose adjacent edges are all horizontal edges of our prior models. Thus {\it each}
adjacent edge is assigned a charge and a spin. They are depicted in the partition functions in (\ref{mybe}).

Each tilted vertex is now assigned a set of Boltzmann weights $\smash{R^{X Y}_{z_1,z_2}}$ depending on a pair of types $X$,~$Y \in \{ \Gamma, \Delta \}$ and a corresponding pair of parameters $z_1$,~$z_2 \in \C^\times$. The $X$ type is associated to the northeast and southwest adjacent edges in the tilted vertex (with associated parameter $z_1$), while the $Y$ type is associated to the northwest and southeast adjacent edges with parameter $z_2$. Thus the order in which these types and parameters are listed matters. We assign Boltzmann weights to each of the four types of tilted vertex according to Table \ref{table:XYice}. As noted previously, any labeling of adjacent edges that is not listed has Boltzmann weight $0$.  

% TABLE OF GAMMA-DELTA, DELTA-DELTA, DELTA-GAMMA, GAMMA-GAMMA
\begin{table}[h]
\begin{threeparttable}
\caption{\parindent=0pt\parskip=0pt The Boltzmann weights for $\Gamma \Delta$, $\Delta \Delta$, $\Delta \Gamma$, and $\Gamma \Gamma$ vertices. For $\Delta \Delta$ and $\Gamma \Gamma$, $a \neq b$ for each tilted vertex whose charges involve only $a$ and $b$. For the last two vertices in the top row of $\Gamma \Delta$, $a \neq b$. The function $g$ satisfies $g(0) = -v$ and $g(n-a) g(a) = v$ if $a \not\equiv 0$ mod $n$.}
\noindent\begin{tabularx}{1.0\textwidth}{@{}l*{5}{>{\centering\arraybackslash}X}@{}}
\toprule
%%%%%%%%%%
%%%%%%%%%%
% GAMMA-DELTA
	$\Gamma\Delta$
	% 1st Config
	&\XYtable{+}{+}{+}{+}{a\vphantom{0}}{0}{a\vphantom{0}}{0}{R^{\Gamma \Delta}_{z_1,z_2}}%
	% 2nd Config
	&\XYtable{-}{-}{-}{-}{0}{a\vphantom{0}}{0}{a\vphantom{0}}{R^{\Gamma \Delta}_{z_1,z_2}}%
	% 3rd Config
	&\XYtable{+}{-}{+}{-}{a\vphantom{0}}{a\vphantom{0}}{a\vphantom{0}}{a\vphantom{0}}{R^{\Gamma \Delta}_{z_1,z_2}}%
	% 4th Config
	&\XYtable{+}{-}{+}{-}{a\vphantom{b}}{b}{a\vphantom{b}}{b}{R^{\Gamma \Delta}_{z_1,z_2}}%
	% 5th Config
	&\XYtable{+}{-}{+}{-}{a\vphantom{b}}{b}{a\vphantom{b}}{b}{R^{\Gamma \Delta}_{z_1,z_2}}%
	\\[1ex]
	{}
	&{\footnotesize$z_{1}^n {-} v z_{2}^n$}
	&{\footnotesize$z_{1}^n {-} v z_{2}^n$}
	&{\scriptsize$(\ast)$}
	&{\footnotesize$v^2 z_{2}^n {-} z_{1}^n$}~{\scriptsize$(\dagger)$}
	&{\scriptsize$(\ddagger)$}
	\\
	\cmidrule(l{0.25cm}){2-6}
	{}
	% 6th Config
	&\XYtable{+}{-}{+}{-}{a\vphantom{b}}{b}{c\vphantom{b}}{d}{R^{\Gamma \Delta}_{z_1,z_2}}%
	% 7th Config
	&\XYtable{-}{+}{-}{+}{0}{0}{0}{0}{R^{\Gamma \Delta}_{z_1,z_2}}%
	% 8th Config
	&\XYtable{-}{+}{+}{-}{0}{0}{a\vphantom{b}}{b}{R^{\Gamma \Delta}_{z_1,z_2}}%
	% 9th Config
	&\XYtable{+}{-}{-}{+}{b}{a\vphantom{b}}{0}{0}{R^{\Gamma \Delta}_{z_1,z_2}}%
	&{}
	\\[2ex]
	{}
	&{\scriptsize$(\S)$}
	&{\footnotesize$z_{1}^n {-} z_{2}^n$}
	&{\footnotesize$(1{-}v) z_{1}^{a} z_{2}^{b-1}$}~{\scriptsize$(\|)$}
	&{\footnotesize$(1{-}v) z_{1}^{a-1} z_{2}^{b}$}~{\scriptsize$(\|)$}
	&{}
	\\
	\midrule	
%%%%%%%%%%
%%%%%%%%%%
% DELTA-DELTA
	$\Delta\Delta$
	% 1st Config
	&\XYtable{+}{+}{+}{+}{0}{0}{0}{0}{R^{\Delta \Delta}_{z_1,z_2}}%
	% 2nd Config
	&\XYtable{-}{-}{-}{-}{a\vphantom{b}}{b}{b}{a\vphantom{b}}{R^{\Delta \Delta}_{z_1,z_2}}%
	% 3rd Config
	&\XYtable{-}{-}{-}{-}{a\vphantom{b}}{b}{a\vphantom{b}}{b}{R^{\Delta \Delta}_{z_1,z_2}}%
	% 4th Config
	&\XYtable{-}{-}{-}{-}{a\vphantom{b}}{a\vphantom{b}}{a\vphantom{b}}{a\vphantom{b}}{R^{\Delta \Delta}_{z_1,z_2}}%
	% 5th Config
	&\XYtable{+}{-}{+}{-}{0}{a\vphantom{b}}{0}{a\vphantom{b}}{R^{\Delta \Delta}_{z_1,z_2}}%
	\\[1ex]
	{}
	&{\footnotesize$z_{1}^n {-} v z_{2}^n$}
	&{\footnotesize$(1{-}v) z_{1}^{n-c} z_{2}^c$}~{\scriptsize$(\#)$}
	&{\footnotesize$g(a{-}b)(z_{1}^n {-} z_{2}^n)$}
	&{\footnotesize$z_{2}^n {-} v z_{1}^n$}
	&{\footnotesize$v(z_{1}^n {-} z_{2}^n)$}
	\\
	\cmidrule(l{0.25cm}){2-6}
	{}
	% 6th Config
	&\XYtable{-}{+}{+}{-}{a\vphantom{b}}{0}{0}{a\vphantom{b}}{R^{\Delta \Delta}_{z_1,z_2}}%
	% 7th Config
	&\XYtable{-}{+}{-}{+}{a\vphantom{b}}{0}{a\vphantom{b}}{0}{R^{\Delta \Delta}_{z_1,z_2}}%
	% 8th Config
	&\XYtable{+}{-}{-}{+}{0}{a\vphantom{b}}{a\vphantom{b}}{0}{R^{\Delta \Delta}_{z_1,z_2}}%
	&{}
	&{}
	\\[2ex]
	{}
	&{\footnotesize$(1{-}v) z_{1}^{n-a+1} z_{2}^{a-1}$}~{\scriptsize$(\ast \ast)$}
	&{\footnotesize$z_{1}^n {-} z_{2}^n$}
	&{\footnotesize$(1{-}v) z_{1}^{a-1} z_{2}^{n-a+1}$}~{\scriptsize$(\ast \ast)$}
	&{}
	&{}
	\\
	\midrule
%%%%%%%%%%
%%%%%%%%%%
% DELTA-GAMMA
	$\Delta \Gamma$
	% 1st Config
	&\XYtable{+}{+}{+}{+}{0}{a\vphantom{0}}{0}{a\vphantom{0}}{R^{\Delta \Gamma}_{z_1,z_2}}%
	% 2nd Config
	&\XYtable{-}{-}{-}{-}{a\vphantom{0}}{0}{a\vphantom{0}}{0}{R^{\Delta \Gamma}_{z_1,z_2}}%
	% 3rd Config
	&\XYtable{+}{-}{+}{-}{0}{0}{0}{0}{R^{\Delta \Gamma}_{z_1,z_2}}%
	% 4th Config
	&\XYtable{-}{+}{-}{+}{a\vphantom{b}}{b}{a\vphantom{b}}{b}{R^{\Delta \Gamma}_{z_1,z_2}}%
	% 5th Config
	&\XYtable{-}{+}{-}{+}{a\vphantom{b}}{b}{a\vphantom{b}}{b}{R^{\Delta \Gamma}_{z_1,z_2}}%
	\\[1ex]
	{}
	&{\footnotesize$z_{2}^n {-} v^n z_{1}^n$}
	&{\footnotesize$z_{2}^n {-} v^n z_{1}^n$}
	&{\footnotesize$z_{2}^n {-} v^{n+1} z_{1}^n$}
	&{\footnotesize$v^{n-1} z_{1}^n {-} z_{2}^n$}~{\scriptsize$(\dagger)$}
	&{\scriptsize$(\dagger \dagger)$}
	\\
	\cmidrule(l{0.25cm}){2-6}
	{}
	% 6th Config
	&\XYtable{-}{+}{+}{-}{b}{a\vphantom{b}}{0}{0}{R^{\Delta \Gamma}_{z_1,z_2}}%
	% 7th Config
	&\XYtable{-}{+}{-}{+}{a}{b}{c\vphantom{d}}{d}{R^{\Delta \Gamma}_{z_1,z_2}}%
	% 8th Config
	&\XYtable{+}{-}{-}{+}{0}{0}{a\vphantom{b}}{b}{R^{\Delta \Gamma}_{z_1,z_2}}%
	&{}
	&{}
	\\[2ex]
	{}
	&{\footnotesize$(1{-}v) v^{a-1} z_{1}^a z_{2}^{b-1}$}~{\scriptsize$(\|)$}
	&{\scriptsize$(\ddagger \ddagger)$}
	&{\footnotesize$(1{-}v) v^{a-1} z_{1}^{a-1} z_{2}^b$}~{\scriptsize$(\|)$}
	&{}
	&{}
	\\
	\midrule
%%%%%%%%%%
%%%%%%%%%%
% GAMMA-GAMMA
	$\Gamma\Gamma$
	% 1st Config
	&\XYtable{+}{+}{+}{+}{a\vphantom{b}}{a\vphantom{b}}{a\vphantom{b}}{a\vphantom{b}}{R^{\Gamma \Gamma}_{z_1,z_2}}%
	% 2nd Config
	&\XYtable{+}{+}{+}{+}{b}{a\vphantom{b}}{b}{a\vphantom{b}}{R^{\Gamma \Gamma}_{z_1,z_2}}%
	% 3rd Config
	&\XYtable{+}{+}{+}{+}{b}{a\vphantom{b}}{a\vphantom{b}}{b}{R^{\Gamma \Gamma}_{z_1,z_2}}%
	% 4th Config
	&\XYtable{-}{-}{-}{-}{0}{0}{0}{0}{R^{\Gamma \Gamma}_{z_1,z_2}}%
	% 5th Config
	&\XYtable{+}{-}{+}{-}{a\vphantom{0}}{0}{a\vphantom{0}}{0}{R^{\Gamma \Gamma}_{z_1,z_2}}%
	\\[1ex]
	{}
	&{\footnotesize$z_{2}^n {-} v z_{1}^n$}
	&{\footnotesize$g(a{-}b) (z_{1}^n {-} z_{2}^n)$}
	&{\footnotesize$(1{-}v) z_{1}^c z_{2}^{n-c}$}~{\scriptsize$(\#)$}
	&{\footnotesize$z_{1}^n {-} v z_{2}^n$}
	&{\footnotesize$v(z_{1}^n {-} z_{2}^n)$}
	\\
	\cmidrule(l{0.25cm}){2-6}
	{}
	% 6th Config
	&\XYtable{-}{+}{-}{+}{0}{a\vphantom{0}}{0}{a\vphantom{0}}{R^{\Gamma \Gamma}_{z_1,z_2}}%
	% 7th Config
	&\XYtable{-}{+}{+}{-}{0}{a\vphantom{0}}{a\vphantom{0}}{0}{R^{\Gamma \Gamma}_{z_1,z_2}}%
	% 8th Config
	&\XYtable{+}{-}{-}{+}{a\vphantom{0}}{0}{0}{a\vphantom{0}}{R^{\Gamma \Gamma}_{z_1,z_2}}%
	&{}
	&{}
	\\[2ex]
	{}
	&{\footnotesize$z_{1}^n {-} z_{2}^n$}
	&{\footnotesize$(1{-}v) z_{1}^a z_{2}^{n-a}$}~{\scriptsize$(\ast \ast)$}
	&{\footnotesize$(1{-}v) z_{1}^{n-a} z_{2}^a$}~{\scriptsize$(\ast \ast)$}
	&{}
	&{}
	\\
	\bottomrule 
\end{tabularx}
% Symbols found in table
\begin{tablenotes}
\vspace{-15pt}
\begin{multicols}{2}
{\footnotesize
\item[*] Weight: $v^2 z_{2}^n -  z_{1}^n$ if $2a \equiv 1 \pod{n}$; else, $g(2a - 1) (z_{1}^n - v z_{2}^n)$.
\item[\textdagger] Here $a + b \equiv 1 \pod{n}$.
\item[\textdaggerdbl] Here $a + b \not\equiv 1 \pod{n}$. Weight: $g(a + b - 1) (z_{1}^n - v z_{2}^n)$.
\item[\S] Here $a + b \equiv c + d \equiv 1 \pod{n}$, $a \not \equiv c \pod{n}$. Let $e \equiv a - c \pod{n}$ with $e \in [0, n - 1]$. Weight: $(v - 1) z_{1}^{n-e} z_2^{e}$ if $ad=0$ or if both $abcd \neq 0$ and $a > c$; $v(v - 1) z_{1}^{n-e} z_2^{e}$ if $bc=0$ or if both $abcd \neq 0$ and $a < c$. 
\item[\parbox{\widthof{$\!$\#}}{$\|$\hfil}\!] Here $a + b \equiv 1 \pod{n}$. Choose $a$ and $b$ in $[1,n]$.
\item[$\!$\#\!] Here $c \equiv a - b \pod{n}$ with $c \in [1,n - 1]$.  
\item[**] Choose $a$ in $[1, n]$. 
\item[\textdagger\textdagger] Here $a + b \not \equiv 1 \pod{n}$. Weight: $(z_{2}^n - v^n z_{1}^n)/g(a + b - 1)$.
\item[\textdaggerdbl\textdaggerdbl] Here $a + b \equiv c + d \equiv 1 \pod{n}$, $a \not \equiv c \pod{n}$. Let $e \equiv c - a \pod{n}$ with $e \in [1, n - 1]$. Weight: $(1 - v) v^{e-1} z_{1}^{e} z_{2}^{n-e}$.
}
\end{multicols}
\end{tablenotes}
\label{table:XYice}
\end{threeparttable}
\end{table}

One can now combine such tilted vertices with vertices of type $R^X_{z_1}$ and $R^Y_{z_2}$ in the rectangular grid discussed in the previous section, in order to create slightly more complicated systems consisting of three vertices. Two examples of such systems are shown below. 

 \vglue -20pt%
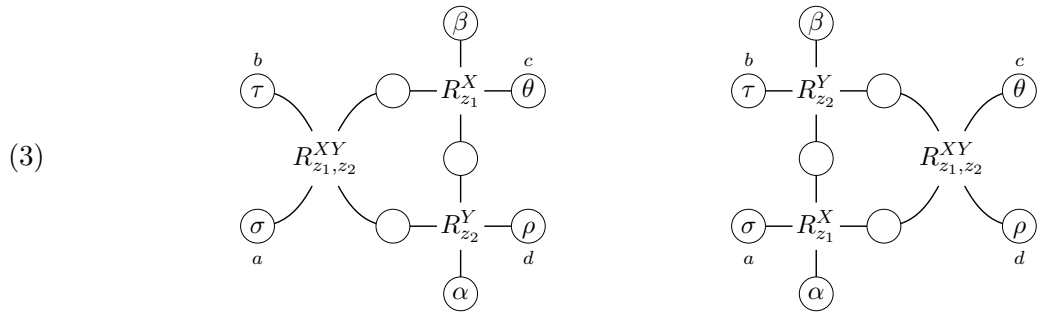
\begin{equation} \label{mybe}
%\begin{array}{c} 
\lhs{\botcharge{\sigma}{a}}{\topcharge{\tau}{b}}{\beta}{\topcharge{\theta}{c}}{\botcharge{\rho}{d}}{\alpha}{\topcharge{}{}}{\botcharge{}{}}{} %\end{array}
%%\qquad \qquad
\hspace{1cm}
%\begin{array}{c} 
\rhs{\botcharge{\sigma}{a}}{\topcharge{\tau}{b}}{\beta}{\topcharge{\theta}{c}}{\botcharge{\rho}{d}}{\alpha}{\botcharge{}{}}{\topcharge{}{}}{} %\end{array}
\end{equation}
 \vglue 15pt%

\begin{definition}\label{definition:ybe}
We say that the triple $[R^{X Y}_{z_1, z_2}, R^{X}_{z_1}, R^{Y}_{z_2}]$ for $X$,~$Y \in \{ \Gamma, \Delta \}$ satisfies the Yang-Baxter equation if the partition functions of the two systems in (\ref{mybe}) are equal for any fixed spins $\alpha$, $\beta$, $\theta$, $\rho$, $\sigma$, $\tau$, and charges $a$, $b$, $c$, $d$.
\end{definition}

The following result is the main tool we use to prove Theorem \ref{theorem:Bprime}. 
\begin{theorem}[\cite{BBB}, \cite{Graythesis}]\label{theorem:ybe}
For any $X$,~$Y \in \{ \Gamma, \Delta \}$, let $R^X$ and $R^Y$ weights be as in Figure~\ref{mweights} and $R^{XY}$ weights as in Table~\ref{table:XYice}. Then the triple $[R^{X Y}_{z_1, z_2}, R^{X}_{z_1}, R^{Y}_{z_2}]$ satisfies the Yang-Baxter equation. 
\end{theorem}
\begin{proof}
Having chosen the Boltzmann weights, the proof is computational. Fix $X$,~$Y \in \{\Gamma, \Delta \}$. Fix boundary spins $\alpha$,~$\beta$,~$\theta$,~$\rho$,~$\sigma$,~$\tau \in \{ +, - \}$ as in (\ref{mybe});  there are $64$ such cases. For each such choice, simply compute the two partition functions for all possible charges $(a,b,c,d)$ mod $n$, and show that they are equal. The case of $X = Y = \Gamma$ was proved in \cite{BBB} and we refer the reader to the Appendix of \cite{Graythesis} where all the computations for the remaining choices of $X$ and $Y$ are done in detail. 
\end{proof}

Using the solution of the Yang-Baxter equation we now prove Theorem \ref{theorem:Bprime}.

\begin{proof}[Proof of Theorem \ref{theorem:Bprime}]
Given $\lambda$, $\mu$, let $\smash{S_{\z, \lambda, \mu}^{\Gamma \Delta}}$ be the resulting two row system. 
Form a new system by attaching an $\smash{R_{z_1,z_2}^{\Delta \Gamma}}$-vertex to the right of $\smash{S_{\z,\lambda, \mu}^{\Gamma \Delta}}$, with the resulting right-hand boundary having both spins $-$, bottom right charge 0, and top right charge arbitrary. Regardless of the choice of top right charge, the only admissible choice for the vertex $\smash{R_{z_1,z_2}^{\Gamma \Delta}}$ has left spins both $-$ and Boltzmann weight $z_2^n-v^n z_1^n$ according to the second column in Table~\ref{table:XYice}. Therefore the partition function of the new system is $(z_2^{n}-v^n z_1^n) Z(\smash{S_{\z,\lambda, \mu}^{\Gamma \Delta}})$. Now apply what Faddeev calls the ``train argument'' (cf. \cite{Baxter, Faddeev}, or \cite{hkice} for depictions closely related to the present context) repeatedly applying the Yang-Baxter equation to move the tilted $R$-vertex from right to left, leaving the partition function unchanged. Once the tilted $R$-vertex has moved all the way to the left, it must have spins which are all $+$ and, according to the first column of Table~\ref{table:XYice}, has weight $z_2^n - v^n z_1^n$, independent of the relevant charges. We conclude that the identity of partition functions in the theorem holds when $z_2^n-v^n z_1^n \neq 0$. If $z_2^n-v^n z_1^n=0$, the result remains true by an easy continuity argument. 
\end{proof}

\section{Yang-Baxter Systems}\label{section:YangBaxterSystem}

\textit{Yang-Baxter systems} are generalizations of the Yang-Baxter equation
that may be used to construct bialgebras
(\cite{FreidelMaillet, HlavatyAlgebraic})
by a generalization of the FRT construction (\cite{FRT}).
We will show that the Yang-Baxter equations that we use in this
paper form a \textit{parametrized Yang-Baxter system}, which
falls in the framework described in \cite{HlavatyApplications}.
Theorem~\ref{byoybs} is then sufficient to construct
a bialgebra from Gamma and Delta ice.

Let us introduce the notation for {\it Yang-Baxter commutators}, the identities of endomorphisms appearing in parametrized Yang-Baxter equations:
 \[ \left\llbracket A, B, C \right\rrbracket = A_{12} (z_1, z_2) B_{13} (z_1,
   z_3) C_{23} (z_2, z_3) - C_{23} (z_2, z_3) B_{13} (z_1, z_3) A_{12} (z_1,
   z_2) \] 
as an endomorphism of the vector space $U \otimes V \otimes W$ with matrices $A \in \text{End}(U \otimes V)$, $B \in \text{End}(U \otimes W)$, and $C \in \text{End}(V \otimes W)$. The subscripts above indicate the factors in $U \otimes V \otimes W$ to which the endomorphism is applied (and is the identity on the third factor). A Yang-Baxter system is a set of four endomorphisms $A$, $B$, $C$, $D$ satisfying the system:
\begin{equation} 
\begin{aligned}
\left\llbracket A, A, A \right\rrbracket = 0, && \left\llbracket A, C, C \right\rrbracket = 0, && \bigl\llbracket A, B^\ddagger, B^\ddagger \bigr\rrbracket = 0,  && \bigl\llbracket A, C, B^\ddagger \bigr\rrbracket = 0, \\
\left\llbracket D, D, D \right\rrbracket = 0, && \left\llbracket D, B, B \right\rrbracket = 0, && \bigl\llbracket D, C^\ddagger, C^\ddagger \bigr\rrbracket = 0, && \bigl\llbracket D, B, C^\ddagger \bigr\rrbracket = 0, \nonumber
\end{aligned}
\end{equation}
where $A^\ddagger(z_1, z_2) := \tau A(z_2, z_1) \tau$ where $\tau$ is the ``flip''
operator $x\otimes y\mapsto y\otimes x$.

\begin{theorem}
  \label{byoybs}
  The matrices $R^{\Gamma \Gamma}, (R^{\Delta
    \Gamma})^{-1}, R^{\Gamma \Delta},$ and $(R^{\Delta \Delta})^\ddagger$ form a Yang-Baxter system with
  $(R^{\Gamma \Delta})^\ddagger = \text{\rm{constant}}\times (R^{\Delta \Gamma})^{-1}$.
\end{theorem}

\begin{proof}
  The identity
\begin{equation}
  \label{secondybe}
  \bigl\llbracket R^{\Gamma \Gamma}, R^{\Gamma \Gamma},
  R^{\Gamma \Gamma} \bigr\rrbracket = 0
\end{equation}
is Theorem~4 of~\cite{BBB}, where it is deduced from the known $R$-matrix of quantum affine
  $\mathfrak{gl}(n|1)$. 
  
  Alternatively, one may deduce the relations needed for the
  Yang-Baxter system from Theorem~\ref{theorem:ybe}.
 For example, pick the configuration on either side of Theorem~4 of~\cite{BBB}
and attach it to the left of a state like the one in Figure
\ref{fig:mice}. One can then use the solutions of Yang-Baxter
equation in Theorem \ref{theorem:ybe} and the train argument to
``move'' the attached ice-system to the right, where it transforms into
the second configuration. This is sufficient to prove
(\ref{secondybe}), which may be deduced by an argument that involves
varying the Langlands parameters. For this argument, see arxiv
version~2 of~\cite{BBB}. This idea applies equally well to the other
Yang-Baxter equations needed for the Yang-Baxter system.

Up to a constant depending on
  $\mathbf{z}$, the fact that $(R^{\Gamma \Delta})^\ddagger$ equals
  $(R^{\Delta \Gamma})^{-1}$, may be checked using the weights given
  in Table~\ref{table:XYice}.
\end{proof}

\bibliographystyle{halpha} \bibliography{statementB.bib}
\end{document}